\documentclass[a4paper, 11pt]{article}

\usepackage{graphics} 
\usepackage{epsfig} 
\usepackage{mathptmx} 
\usepackage{times} 
\usepackage{amsthm}
\usepackage{amsmath} 
\usepackage{amssymb, color}  
\usepackage{stmaryrd}
\newtheorem{theo}{Theorem}
\newtheorem{lem}{Lemma}

\newtheorem{defi}{Definition}
\newtheorem{propo}{Proposition}
\newtheorem{remark}{Remark}

\newcommand{\rref}[1]{(\ref{#1})}

\newcommand{\scal}{\mathcal{S}}

\newcommand{\pcal}{\mathcal{P}}

\newcommand{\abs}[1]{\left| #1 \right|}
\newcommand{\reels}{\mathbb{R}}

\newcommand{\sign}{\text{sign}}

\title{\Large\bf
Global stabilization of multiple integrators by a bounded feedback with constraints on its successive derivatives
}

\author{Jonathan Laporte, Antoine Chaillet and Yacine Chitour 
\thanks{This research was partially supported by a public grant overseen by the French ANR as part of the “Investissements d'Avenir” program, through the iCODE institute, research project funded by the IDEX Paris-Saclay, ANR-11-IDEX-0003-02.}
\thanks{J. Laporte, A. Chaillet  and  Y. Chitour  are with L2S - Univ. Paris Sud - CentraleSup\'elec. 3, rue Joliot-Curie. 91192 - Gif sur Yvette, France.
        {\tt\small jonathan.laporte, antoine.chaillet, yacine.chitour@l2s.centralesupelec.fr}}%
}

\begin{document}
\newcommand{\AC}[1]{\textbf{\textcolor{red}{#1}}}
\newcommand{\JL}[1]{\textbf{\textcolor{green}{#1}}}
\date{}

\maketitle

\begin{abstract}
In this paper, we address the global stabilization of chains of integrators by means of a bounded static feedback law whose $p$ first time derivatives are bounded. Our construction is based on the technique of nested saturations introduced by Teel.  We show that the control amplitude and the maximum value of its $p$ first derivatives can be imposed below any prescribed values. Our results are illustrated by the stabilization of the third order integrator on the feedback and its first two derivatives.
\end{abstract}

\section{Introduction}

Actuator constraints is an important practical issue in control applications since it is a possible source of instability or performance degradation. Global stabilization of linear time-invariant (LTI) systems with actuator saturations (or bounded inputs) can be achieved if and only the uncontrolled linear system has no eigenvalues with positive real part and is stabilizable \cite{SSY}. 

Among those systems, chains of integrators have received specific attention. Saturation of a linear feedback is not globally stabilizing as soon as the integrator chain is of dimension greater than or equal to three \cite{FULLER69,SY91}. In \cite{Teel92} a globally stabilizing feedback is constructed using {\it nested saturations} for the multiple integrator. This construction has been extended to the general case in \cite{SSY}, in which a family of stabilizing feedback laws is proposed as a linear combination of saturation functions. In \cite{Marchand2003b} and \cite{Marchand2005}, the issue of performance of these bounded feedbacks is investigated for multiple integrators and some improvements are achieved by using 
variable levels of saturation. A gain scheduled feedback was proposed in \cite{Megretski96bibooutput} to ensure robustness to some classes of bounded disturbances. Global practical stabilization has been achieved in \cite{Gayaka:2011bm} in the presence of bounded actuator disturbances using a backstepping procedure.

Technological considerations may not only lead to a limited amplitude of the applied control law, but also to a limited reactivity. This problem is known as rate saturation \cite{lauvdal97} and corresponds to the situation when the signal delivered by the actuator cannot have too fast variations.
This issue has been addressed for instance in \cite{SilvaTarbouch03}-\cite{Freeman:1998tp}. In \cite{SilvaTarbouch03, Galeani}, regional stability is ensured through LMI-based conditions. In \cite{lauvdal97}, a gain scheduling technique is used to ensure semi-global stabilization of integrator chains. In \cite{saberi2012}, semi-global stabilization is obtained via low-gain feedback or low-and-high-gain feedback. In \cite{Freeman:1998tp}, a backstepping procedure is proposed to globally stabilize a nonlinear system with a control law whose amplitude and first derivative are bounded independently of the initial state. 

In this paper, we deepen the investigations on global stabilization of LTI systems subject to bounded actuation with rate constraints. We consider rate constraints that affect only the first derivative of the control signal, but also its successive $p$ first derivatives, where $p$ denotes an arbitrary positive integer. Focusing on chains of integrators of arbitrary dimension, we propose a static feedback law that globally stabilizes chains of integrators, and whose magnitude and $p$ first derivatives are below arbitrarily prescribed values at all times. Our control law is based on the nested saturations introduced in \cite{Teel92}. We rely on specific saturation functions, which are linear in a neighborhood of the origin and constant for large values of their argument. 

This paper is organized as follows. In Section \ref{section:stat_main_res}, we provide definitions and state our main result. The proof of the main result is given in Section \ref{sec:proof_main_result} based on several technical lemmas. In Section \ref{sec:simu}, we test the efficiency of the proposed control law via numerical simulations on the third order integrator, with a feedback whose magnitude and two first derivatives are bounded by prescribed values. We provide some conclusions and possible future extensions in Section \ref{sec: conclusion}.

\vspace{3mm}
\textbf{Notations.}
The function $\sign:\reels \backslash \lbrace{0\rbrace}\to \mathbb R$ is defined as $\sign(r) := r / \abs{r}$. Given a set $I\subset \mathbb R$ and a constant $a\in\mathbb R$, we let $I_{\geq a}:=\left\{x\in I\,:\, x\geq a\right\}$.  Given $k\in\mathbb N$ and $m\in\mathbb N_{\geq 1}$, we say that a function $f : \reels^m \rightarrow \reels$ is of class $C^{k}(\reels^m , \reels)$ if its differentials up to order $k$ exist and are continuous, and we use $f^{(k)}$ to denote the $k$-th order derivative of $f$. By convention, $f^{(0)}:=f$. The factorial of $k$ is denoted by $k!$. We define $\llbracket m , k \rrbracket:=\left\{n\in\mathbb N\,:\, n\in[m,k]\right\}$. We use $\reels^{m,m}$ to denote the set of $m\times m$ matrices with real coefficients.  $J_m\in \reels^{m,m}$ denotes the $m$-th Jordan block, i.e. the $m \times m$ matrix given by $(J_m)_{i,j} =1$ if $i=j-1$ and zero otherwise. For each $i\in\llbracket 1 , m \rrbracket$, $e_i\in \reels^m$ refers to the column vector with coordinates equal to zero except the $i$-th one equal to one.


\section{Statement of the main result}
\label{section:stat_main_res}

In this section we present our main result on the stabilization of the multiple integrators with a control law whose magnitude and $p$ first derivatives are bounded by prescribed constants. Given $n\in\mathbb N_{\geq 1}$, the multiple integrator of length $n$ is given by
\begin{equation}
\label{mult_int}
  \left\{
      \begin{array}{r c  l}
        \dot{x}_1& = & x_2, \\
        &  \vdots& \\
        \dot{x}_{n-1} & =&  x_{n}, \\
        \dot{x}_n&= & u.
      \end{array}
    \right.
\end{equation}
Letting $x:=(x_1,\ldots,x_n)$, System \rref{mult_int} can be compactly written as 
$$
\dot{x}= J_n x +e_n u.
$$
In order to make the objectives of this paper more precise, we start by introducing the notion of $p$-bounded feedback law by $(R_j)_{0 \leq j \leq p}$ for System \rref{mult_int}, which will be used all along the document.

\vspace{3mm}
\begin{defi}
Given $n\in\mathbb N_{\geq 1}$ and $p\in\mathbb N$, let $(R_j)_{0 \leq j \leq p}$ denote a family of positive constants. We say that $\nu : \reels^n \rightarrow \reels $ is a \textit{$p$-bounded feedback law by $(R_j)_{0 \leq j \leq p}$  for System \rref{mult_int}} if, for every trajectory of the closed loop system $\dot{x}= J_n x +e_n \nu(x)$, the time function $u : \reels_{\geq 0} \rightarrow  \reels $ defined by $u(t) = \nu (x(t)) $ for all $t \geq 0$ satisfies, for all $j \in \llbracket 1 , p \rrbracket$,
\begin{equation*}
\label{est:def1}
\sup\limits_{t\geq 0} \left\lbrace \abs{u^{(j)}(t)} \right\rbrace \leq R_j .
\end{equation*}
\end{defi}
\vspace{3mm}

Based on this definition, we can restate our stabilization problem as follows. Given $p\in\mathbb N$ and a set of positive real numbers $(R_j)_{0 \leq j \leq p}$, our aim is to design a feedback law $\nu $  which is a $p$-bounded feedback law by $(R_j)_{0 \leq j \leq p}$ for System \rref{mult_int} such that the origin of the closed-loop system $\dot{x}= J_n x +e_n \nu(x)$ is globally asymptotically stable. The case $p=0$ corresponds to global stabilization with bounded state feedback and has been addressed in e.g. \cite{Teel92,Marchand2003b,Marchand2005}. The case $p=1$ corresponds to global stabilization with bounded state feedback and limited rate, in the line of e.g. \cite{SilvaTarbouch03, Galeani, lauvdal97, saberi2012, Freeman:1998tp}. Inspired by \cite{Teel92}, our design for an arbitrary order $p$ is based on a nested saturations feedback, where saturations belong to the following class of functions. 

\vspace{3mm}
\begin{defi}
\label{def:S(p)}
Given $p\in\mathbb N$, $\scal (p)$ is defined as the set of all functions $ \sigma $ of class $ C^{p}(\reels , \reels)$, which are odd, and such that there exists positive constants  $\alpha$, $L$, $\sigma^{max}$ and $S$ satisfying, for all $r \in \reels$,
\begin{itemize}
\item[(i)] $r \sigma (r) > 0$, when $r \neq 0$,
\item[(ii)] $\sigma (r)= \alpha r $, when $\abs{r} \leq L$,
\item[(iii)]  $\abs{\sigma (r)} =\sigma^{max} $, when $\abs{r} \geq S$.
\end{itemize}
In the sequel, we associate with every \emph{$\sigma \in \scal(p)$ the $4$-tuple $(\sigma^{max},L,S,\alpha)$}.
\end{defi}

\vspace{3mm}

The constants $\sigma^{max}$, $L$, $\alpha$, and $S$ will be extensively used throughout the paper. Figure \ref{fig:S(p)} helps fixing the ideas. $\sigma^{max}$ represents the saturation level, meaning the maximum value that can be reached by the saturation. $L$ denotes the linearity threshold: for all $|r|\leq L$, the saturation behaves like a purely linear gain. $\alpha$ is the value of this gain that is, the slope of the saturation in the linear region. $S$ represents the saturation threshold: for all $|r|\geq S$, the function saturates and takes a single value (either $-\sigma^{max}$ or $\sigma^{max}$). Notice that it necessarily holds that $S \geq L$  and the equality may only hold when $p = 0$. We also stress that the successive derivatives up to order $p$ of an element of $\scal(p)$ are bounded. An example of such function is given in Section \ref{sec:simu} for $p=2$. 

\begin{figure}[thpb]
      \centering
      \includegraphics[scale=0.9]{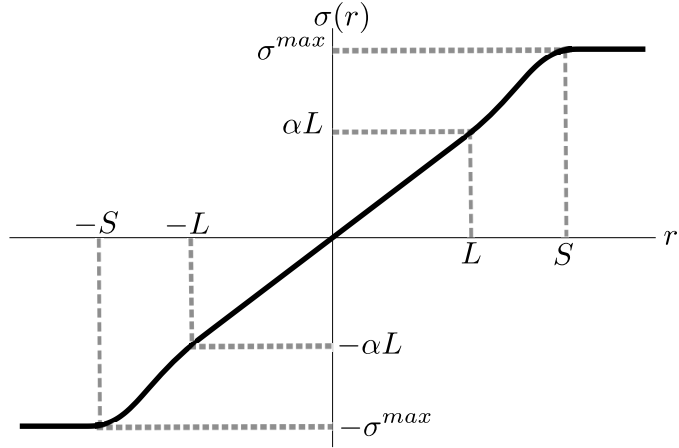}
      \caption{A typical example of a $\scal(p)$ saturation function with constants $(\sigma^{max},L,S,\alpha)$.}
      \label{fig:S(p)}
\end{figure}

Based on these two definitions, we are now ready to present our main result, which establishes that global stabilization on any chain of integrators by bounded feedback with constrained $p$ first derivatives can always be achieved by a particular choice of nested saturations.

\vspace{3mm}
\begin{theo}
\label{Main_res}
Given $n\in\mathbb N_{\geq 1}$ and $p\in\mathbb N$, let $(R_j)_{0 \leq j \leq p}$ be a family of positive constants. For every set of saturation functions  $ \sigma_1 , \ldots , \sigma_n \in \scal(p)$, there exists  vectors $ k_1 , \ldots , k_n $ in $\reels^n$, and positive constants  $ a_1 , \ldots , a_n $ such that the feedback law $\nu$ defined, for each $x\in\mathbb R^n$, as
\begin{align}
\label{nested_com_th}
\nu (x)= - a_n \sigma_n \Big( k_n^T x+a_{n-1} \sigma_{n-1}\big(k_{n-1}^T x +\ldots  + a_1  \sigma_1 (k_1^T x )\big) \ldots \Big)
\end{align}
is a $p$-bounded feedback law by $(R_j)_{0 \leq j \leq p}$ for System \rref{mult_int}, and the origin of the closed-loop system $\dot{x}= J_n x +e_n^T \nu(x)$  is globally asymptotically stable.
\end{theo}
\vspace{3mm}

The proof of this result is given in Section \ref{sec:proof_main_result}. It provides above theorem we give below also provides an explicit choice of the gain vectors $ k_1, \ldots, k_n$ and constants $a_1,\ldots a_n $.

\vspace{3mm}
\begin{remark}
In \cite{SSY}, a stabilizing feedback law was constructed using linear combinations of saturated functions. That feedback with saturation functions in $\scal(p)$ cannot be a $p$-bounded feedback for System \rref{mult_int}. To see this, consider the multi-integrator of length $2$, given by $\dot{x}_1=x_2, \: \dot{x}_2=u$. Any stabilizing feedback using a linear combination of saturation functions in $\scal (p)$ is given by $\nu (x_1,x_2) = - a  \sigma_1(b x_2) - c \sigma_2(d (x_2 + x_1))$, where the constants $a$, $b$, $c$, and $d$ are chosen to insure stability of the closed-loop system according to \cite{SSY}. Let $u(t)= \nu( x_1(t),x_2(t))$ for all $t \geq 0$. A straightforward computation yields $\dot{u}(t) =- a b \sigma_{1}^{(1)}(a x_2 (t))  u(t) - c  d \sigma^{(1)}_{2}(d (x_2(t) + x_1(t))) (x_2(t) + u(t))$. Now consider a solution with initial condition $x_2(0)=x_{20}$, and $x_1(0)=-x_{20}$ such that $\sigma_{1}^{(1)}(a x_{20})=0$. We then have $\dot{u}(0)  =  -cd \sigma^{(1)}_2(0) (x_{20} + u(0))$, whose norm is greater than $A(\vert x_{20}\vert -B)$ for some positive constants $A,B$. Thus $|\dot u(0)|$ grows unbounded as $|x_{20}|$ tends to infinity, which contradicts the definition of a $p$-bounded feedback.
\end{remark}
\vspace{3mm}

\begin{remark}
Our construction is developed for chains of integrator, but it may fails for a general linear system stabilizable by bounded inputs. Consider for instance the harmonic oscillator given by $\dot{x}_1= x_2$, $\dot{x}_2 = -x_1 + u $ and a bounded stabilizing law given by $u(t) =- \sigma ( x_2(t))$ with $\sigma \in \scal (p)$ for some integer $p$. The time derivative of $u$ verifies $\abs{\dot{u}(t)} \geq  \abs{  \sigma^{(1)}(x_2(t))}(\abs{x_1(t)} - \abs{u(t)})$, which grows unbounded as the state norm increases, thus contradicting the definition of $p$-bounded feedback.
\end{remark}
\vspace{3mm}

\section{Proof of the main result}
\label{sec:proof_main_result}
\subsection{Technical lemma}
\label{sec:tech:lem}
We start by giving a lemma that provides an upper bound of composed functions by exploiting the saturation region of the functions in $\scal (p)$.

\vspace{3mm}
\begin{lem}
\label{lem:techn}
Given $k\in\mathbb N$, let $f$ and $g$ be functions of class $C^k(\reels_{\geq 0} , \reels)$, $\sigma$ be a saturation function in $\scal (k)$ with constants ($\alpha ,L ,S , \sigma^{max}$), and $E$ and $F$ be subsets of $\reels_{\geq 0}$ such that $E \subseteq F$. Assume that
\begin{eqnarray}
\label{eq:lem:tech:absf}
& \: \abs{f(t)} > S ,& \:  \quad  \forall t \in F \backslash E,
\end{eqnarray} and there exists positive constants $M, Q_1, \ldots , Q_k$ such that
\begin{eqnarray}
\label{eq:lem:tech:fkQ}
 & \abs{f^{(k_1)}(t)} \leq Q_{k_1}, &  \quad \: \forall t \in E, \: \forall  k_1 \in \llbracket 1 , k \rrbracket , \\ \label{eq:lem:tech:gkM}
 &  \:\abs{g^{(k)}(t)} \leq M, & \quad  \forall t \in F.
\end{eqnarray} Then the $k$th-order derivative of $h : \reels_{\geq 0} \rightarrow \reels$, defined by $h(\cdot) = g(\cdot) + \sigma ( f(\cdot))$, satisfies
\begin{equation}
\label{lem:eq:est}
\abs{h^{(k)}(t)} \leq M + \sum\limits_{a=1}^{k}  \overline{\sigma}_a B_{k,a}(Q_1, \ldots , Q_{k -a+1})  ,  \quad \forall  t   \in F , 
\end{equation}
where $B_{k,a}(Q_1, \ldots , Q_{k -a+1})$ is a polynomial function of $Q_1, \ldots , Q_{k -a+1}$, and $\overline{\sigma}_a := \max_{s \in \reels} |\sigma^{(a)}(s)|$ for each $a\in \llbracket 1 , k \rrbracket$.
\end{lem}
\vspace{3mm}

\begin{proof}[Proof of Lemma \ref{lem:techn}]
The proof relies on Fa\`a Di Bruno's formula, which we recall

\vspace{3mm}
\begin{lem}[Fa\`a Di Bruno's formula, \cite{fdb}, p. 96]
\label{lem:fa_di}
Given  $k\in\mathbb N$, let $\phi\in C^{k}( \reels_{\geq 0} , \reels )$ and $\rho\in C^{k}( \reels , \reels )$. Then the $k$-th order derivative of the composite function $\rho \circ \phi$ is given by
\begin{equation}
\label{eq:faadibruno}
\frac{d^k}{dt^k}\rho(\phi(t)) =  \sum\limits_{a=1}^k \rho^{(a)} (\phi(t)) B_{k,a}\Big(\phi^{(1)}(t), \ldots , \phi^{(k-a+1)}(t)\Big),
\end{equation}
where $B_{k,a}$ is the Bell polynomial given by
\begin{align}
\label{bell}
B_{k,a}\Big(\phi^{(1)}(t), \ldots ,  \phi^{(k-a+1)}(t)\Big)\hspace{-1mm}:=\hspace{-2mm}\sum\limits_{\delta \in \pcal_{k,a}} \hspace{-1mm}c_{\delta} \hspace{-1mm}\prod\limits_{l=1}^{k-a+1} \left( \phi^{(l)}(t) \right)^{\delta_l}
\end{align}
where $\pcal_{k,a}$ denotes the set of $(k-a+1)-$tuples $\delta :=(\delta_1 , \delta_2, \ldots , \delta_{k-a+1})$  of positive integers satisfying
\begin{align*}
\delta_1 + \delta_2 + \ldots +\delta_{k-a+1} &= a,\\
 \delta_1 +2 \delta_2 + \ldots + (k-a+1) \delta_{k-a+1} &= k,
\end{align*}
and $c_{\delta}:=k!/\left(\delta_1 ! \cdots \delta_{k-a+1}! (1!)^{\delta_1} \cdots ((k-a+1)!)^{\delta_{k-a+1}}\right)$.
\end{lem}
\vspace{3mm}

Using Lemma \ref{lem:fa_di}, a straightforward computation yield
\begin{equation*}
h^{(k)}(t) = g^{(k)}(t) + \sum\limits_{a=1}^{k} \sigma^{(a)} (f(t))  B_{k,a}\Big(f^{(1)}(t), \ldots , f^{(k-a+1)}(t)\Big).
\end{equation*} 
Since $\sigma \in \scal(k)$, \rref{eq:lem:tech:absf} ensures that the set $F\setminus E$ is contained in the saturation zone of $\sigma$. It follows that
\begin{equation}
\label{proof:lem:tech:int}
\frac{d^{k}}{dt^{k}} \sigma(f(t))=0,\quad \forall t\in F\setminus E.
\end{equation}
Furthermore, from \rref{eq:lem:tech:fkQ} and \rref{bell} it holds that, for all $t \in E$,
\begin{eqnarray*}
\abs{B_{k,a}\left( f^{(1)}(t), \ldots , f^{(k -a+1)}(t)\right)}
 &\leq &\sum\limits_{\delta \in \pcal_{k,a}} c_{\delta} \prod\limits_{l=1}^{k-a+1} Q_l^{\delta_l}, \\
& = & B_{k,a}(Q_1, \ldots ,Q_{k-a+1}).
\end{eqnarray*}
From definition of $\overline{\sigma}_a $ and \rref{eq:faadibruno}, we get that
\begin{equation}
\label{est_tech_1}
 \abs{\frac{d^{k}}{dt^k}  \sigma(f(t))} \leq  \sum\limits_{a=1}^{k} \overline{\sigma}_a B_{k,a}(Q_1, \ldots ,Q_{k -a+1}), \quad \forall t \in E.
\end{equation} 
In view of \rref{proof:lem:tech:int}, the estimate \rref{est_tech_1} is valid on the whole set $F$. Thanks to \rref{eq:lem:tech:gkM}, a straightforward computation leads to the estimate \rref{lem:eq:est}.
\end{proof}

\subsection{Intermediate results}
\label{sec:int:res}
In this subsection we provide two propositions which will be used in the proof of Theorem \ref{Main_res}. We start by introducing some necessary notation.

Given $n\in\mathbb N_{\geq 1}$ and $p\in\mathbb N$, let $ \mu_1 , \ldots , \mu_n$ be saturations in $\scal(p)$ with respective constants $(\mu_i^{max},  L_{\mu_i} ,S_{\mu_i} , \alpha_{\mu_i} ) $, $i \in \llbracket  1, n   \rrbracket$. We define, for each $i \in \llbracket  1, n   \rrbracket$,
\begin{align}
\label{def_up_bnd_et}
\overline{\mu}_{i,j} & :=  \max  \left\lbrace \abs{\mu_i^{(j)}(r)}: \:  r \in  \reels  \right\rbrace , \quad \forall j \in \llbracket  1, p  \rrbracket , \\ \label{eq:b_mui}
b_{\mu_i} & :=  \max \left\lbrace \abs{r - \mu_i(r)} : \:\abs{r} \leq S_{\mu_i} + 2 \mu_{i-1}^{max}  \right\rbrace  .
\end{align}
We also let
\begin{align} \label{def_bound_mu_n_sup}
\overline{b}_{\mu_n} & :=  \max \left\lbrace \frac{\mu_n(r)}{r} : \: 0 < \abs{r} \leq S_{\mu_n} \right\rbrace ,\\ \label{def_bound_mu_n_inf}
\underline{b}_{\mu_n} & := \min \left\lbrace \frac{\mu_n(r)}{r} : \: 0 < \abs{r} \leq S_{\mu_n} \right\rbrace . 
\end{align}
Note that these quantities are well defined since the functions $\mu_i$ are all in $\scal (p)$. 

We also make a linear change of coordinates $y = H x $, with $H\in \reels^{n,n}$, that puts System \rref{mult_int} into the form
\begin{equation}
\label{S1}
\dot{y}_i  =  \alpha_{\mu_n}  \sum\limits_{l=i+1}^n y_l + u, \quad \forall i \in \llbracket 1 , n \rrbracket ,
\end{equation}
with the convention $\sum\limits_{l=n+1}^n=0$. The matrix $H$ can be determined from
\begin{equation}
y_{n-i}=\sum\limits_{k=0}^i \frac{i!}{k!(i-k)!} \left( \alpha_{\mu_n}  \right)^k  x_{n-k}, \quad \forall i \in \llbracket 0, n-1 \rrbracket.
\end{equation} 
For this system, we define a nested saturations feedback law $\Upsilon : \reels^n \rightarrow \reels $ as
\begin{equation}
\label{fe:prop:S_1}
\Upsilon(y)=- \mu_{n} (y_n + \mu_{n-1}(y_{n-1}+ \ldots + \mu_1(y_1))\ldots ). 
\end{equation}
Let $y(\cdot)$ be a trajectory of the system
\begin{equation}
\label{S1cc}
\dot{y}_i  =  \alpha_{\mu_n}  \sum\limits_{l=i+1}^n y_l + \Upsilon(y), \quad \forall i \in \llbracket 1 , n \rrbracket ,
\end{equation}
which is the closed-loop system \rref{S1} with the feedback defined in \rref{fe:prop:S_1}. For each $i\in  \llbracket 1 , n\rrbracket$, the time function $z_i : \reels_{\geq 0} \rightarrow \reels$ is defined recursively as
\begin{equation*}
\label{s_i_def}
z_i(\cdot ): = y_i(\cdot) + \mu_{i-1}(s_{i-1}(\cdot)),
\end{equation*}
with $\mu_0(\cdot)= 0$. Notice that with the above functions, the closed loop system \rref{S1cc} can be rewritten as
\begin{equation}
\begin{cases}
\label{S1c2}
\dot{y}_i  = \alpha_{\mu_n} z_n - \mu_{n} (z_n) + \alpha_{\mu_n}  \sum\limits_{l=i+1}^{n-1}( z_l - \mu_{l} (z_l )) -\alpha_{\mu_n} \mu_i(z_i), \quad \forall i \in \llbracket 1 , n-1 \rrbracket ,\\
\dot{y}_n  =  -  \mu_{n} (z_n).
\end{cases}
\end{equation}
For $i\in  \llbracket 1 , n\rrbracket$, we also let
\begin{equation}
\label{E_i_def}
E_i  :=  \left\lbrace y \in \reels^{n} : \:  \abs{y_v } \leq S_{\mu_v} + \mu_{v-1}^{max} , \forall v \in \llbracket i, n \rrbracket  \right\rbrace ,
\end{equation}
with $\mu_{0}^{max}=0$, and 
\begin{equation}
\label{I_i_def}
 I_i :=   \lbrace t \in \reels_{\geq 0} : \: y(t) \in E_i \rbrace .
\end{equation}
Note that from the definitions of $I_i$ and $E_i$, we have $I_{1} \subseteq I_{2} \subseteq \ldots \subseteq I_n$, and a straightforward computation yields
\begin{align}
\abs{z_{i}(t)} &> S_{\mu_{i}} , \quad \forall t \in I_{i+1} \backslash I_{i}, \:   \forall i \in \llbracket 1 , n-1 \rrbracket , \label{eq:s_i_I_i_prive_I_i+1}\\
\abs{z_n(t)} &> S_{\mu_{n}} , \quad \forall t \in \reels_{\geq 0} \backslash I_{n},\label{eq:s_n_rn}
\end{align}
which allows us to determine when saturation occurs. Moreover from the definitions of saturation functions of class $\scal(p)$, $E_i$, $I_i$, \rref{def_bound_mu_n_sup} and \rref{def_bound_mu_n_inf}, the following estimates can easily be derived:
\begin{align}
\label{eq:si-sig}
\abs{z_i(t) - \mu_i(z_i(t))} &\leq b_{\mu_i} , \quad \forall t \in I_i,
\\
\abs{\alpha_{\mu_n} z_n(t) \hspace{-0.5mm}-\hspace{-0.5mm}\mu_n (z_n(t))} &\leq (\overline{b}_{\mu_n}\hspace{-0.5mm} - \hspace{-0.5mm}\underline{B}_{\mu_n}) (S_{\mu_{n}} \hspace{-0.5mm}+\hspace{-0.5mm}  2 \mu_{n-1}^{max} )  , \quad \forall t \in I_n, \label{eq:sn-sig}
\end{align}
with $\underline{B}_{\mu_n} := \min \left\lbrace \underline{b}_{\mu_n}, \frac{\mu_n^{max}}{S_{\mu_{n}} + 2 \mu_{n-1}^{max}} \right\rbrace$.

The following statement provides explicit bounds on the successive derivatives of each functions $y_i(t)$, $z_i(t)$ for each $i  \in \llbracket 1, n \rrbracket$ and the time function given by $u(\cdot)=\Upsilon(y(\cdot))$.

\vspace{3mm}
\begin{propo}
\label{prop:2}
Given $n\in\mathbb N_{\geq 1}$ and $p\in\mathbb N$, let $ \mu_1 , \ldots , \mu_n$ be saturation functions in $\scal(p)$ with respective constants $(\mu_i^{max},  L_{\mu_i} , S_{\mu_i} , \alpha_{\mu_i} ) $ for each $i  \in \llbracket 1, n \rrbracket $. With the notation introduced in this section and the Bell polynomials introduced in \eqref{bell}, every trajectory of the closed-loop system \rref{S1cc} satisfies, for each $i\in\llbracket 1,n\rrbracket$ and each $j \in \llbracket 1, p \rrbracket $,
\begin{align}
(P_1(i,j)) :\quad & \: \abs{y_i^{(j)}(t)} \leq Y_{i,j} ,\quad \forall t \in I_i \,;  \label{y_j_est}\\ 
(P_2(i,j)) :\quad & \: \abs{z_i^{(j)}(t)} \leq  Z_{i,j} , \quad \forall t \in I_i \,   ; \label{s_i_j_est}\\
(P_3(j)) :\quad & \:\sup\limits_{t\geq 0} \left\lbrace \abs{u^{(j)}(t)} \right\rbrace   \leq   \sum\limits_{q=1}^j  G_{q,j} \overline{\mu}_{n,q} \,;\label{sup_u_j_est}
\end{align}
where $Y_{i,j}$, $Z_{i,j}$, and $G_{q,j}$ are independent of initial conditions and are obtained recursively as follows: for $j=1$,
\begin{align*}
Y_{n,1} &:= \mu_n^{max}, \\
Y_{i,1} & := (\overline{b}_{\mu_n} - \underline{B}_{\mu_n})(S_{\mu_{n}} + 2 \mu_{n-1}^{max}) +  \alpha_{\mu_n} \sum\limits_{l=i+1}^{n-1} b_{\mu_l} + \alpha_{\mu_n}  \mu_i^{max},\quad \forall i \in \llbracket 1, n-1 \rrbracket ,\\
Z_{1,1} & :=  Y_{1,1}, \\
Z_{i,1} & :=  Y_{i,1} + \overline{\mu}_{i-1,j} Z_{i-1,1},\quad \forall i \in \llbracket 2, n \rrbracket , \\
G_{1,1} & := Z_{n,1}
\end{align*}
and, for each $j\in\llbracket 2, p\rrbracket$,
\begin{align*} 
Y_{i,j} & :=  \alpha_{\mu_n} \sum\limits_{b=i+1}^{n} Y_{b,j-1} + \sum\limits_{q=1}^{j-1} G_{q,j-1} \overline{\mu}_{n,q} , \quad \forall i \in \llbracket 1, n-1 \rrbracket ,\\
Z_{1,j} & :=  Y_{1,j},\\ 
Z_{i,j} & :=  Y_{i,j} + \sum\limits_{a=1}^j \overline{\mu}_{i-1,a} B_{j,a}(Z_{i-1,1}, \ldots , Z_{i-1,j-1+a}), \quad \forall i \in \llbracket 2, n \rrbracket ,\\
G_{q,j} & := B_{j,q}(Z_{n,1}, \ldots , Z_{n,j-q+1}),  \quad \forall q \in \llbracket 1, j \rrbracket .
\end{align*}
\end{propo}

\vspace{3mm}
\begin{proof}[Proof of Proposition \ref{prop:2}]

Let $y(t)$ be a trajectory of the closed loop system \rref{S1cc}. The right-hand side of \rref{S1cc} being of class $C^p( \reels^n , \reels^n )$ and globally Lipschitz, System \rref{S1cc} is forward complete and its trajectories are of class $C^{p+1}( \reels_{\geq 0} , \reels^n )$. Therefore the successive time derivatives of $y_i(t)$, $z_i(t)$, and $u(t)$ are well defined.

We establish the result by induction on $j$. We start by $j=1$. We begin to prove that $P_1(i,1)$  holds for all $i \in \llbracket 1, n \rrbracket $. Let $i \in \llbracket 1, n-1 \rrbracket $. From \rref{S1c2}, \rref{eq:si-sig}, and \rref{eq:sn-sig} a straightforward computation leads to
\begin{equation*}
\abs{ \dot{y}_i(t) } \leq  (\overline{b}_{\mu_n} - \underline{B}_{\mu_n})(S_{\mu_{n}} + 2 \mu_{n-1}^{max})+   c \sum\limits_{l=i+1}^{n-1} b_{\mu_l} + c  \mu_i^{max}, 
\end{equation*}
for all $t \in I_{i+1}$. Since $I_{i} \subseteq I_{i+1}$, the above estimate is still true on $I_i$. Moreover, from \rref{S1c2} it holds that $\abs{ \dot{y}_n(t) } \leq \mu_n^{max}$ at all positive times. $P_1(i,1)$ has been proven  for each $i \in \llbracket 1, n \rrbracket $. 

We now prove by induction on $i$ the statement $P_2(i,1)$. Since $z_1(\cdot )=y_1(\cdot )$, the case $i=1$ is done. Assume that, for a given $i \in \llbracket 1, n-1 \rrbracket $, the statement $P_2(i_2, 1)$ holds for all $i_2 \leq i$. From Lemma \ref{lem:techn} (with $k=1$, $f=z_i$, $g=y_{i+1}$, $h=z_{i+1}$, $\sigma = \mu_i$, $Q_1= Z_{i,1}$, $M = Y_{i+1,1}$, $\overline{\sigma}_1 = \overline{\mu}_{i,1}$, $E=I_i$, $F=I_{i+1}$, and \rref{eq:s_i_I_i_prive_I_i+1}), we can establish that $P_2(i+1,1)$ holds. Thus $P_2(i,1)$ holds for all $i \in \llbracket 1, n \rrbracket $.

Notice that $u(\cdot)=- \mu_n(z_{n}(\cdot))$. We then can establish $P_3(1)$ from Lemma \ref{lem:techn} (with $k=1$, $f=z_n$, $g\equiv 0$, $h=u$, $\sigma = \mu_n$, $Q_1= Z_{1,i}$, $M = 0$, $\overline{\sigma}_1 = \overline{\mu}_{n,1}$, $E=I_n$, $F=\reels_{\geq 0}$ and \rref{eq:s_n_rn}). This ends the case $j=1$.

Assume that for a given $j \in \llbracket 1, p-1 \rrbracket$, statements $P_1(i,j_2)$, $P_2(i,j_2)$ and $P_3(j_2)$ hold for all $j_2 \leq j$ and all $i \in \llbracket 1, n \rrbracket $. Let $i \in \llbracket 1, n \rrbracket $. From \rref{S1}, a straightforward computation yields
\begin{eqnarray*}
\abs{ y_i^{(j+1)}(t) }  \leq  \alpha_{\mu_n} \sum\limits_{l=i+1}^{n} \abs{ y_l^{(j)}(t) } + \abs{u^{(j)}(t)}, \quad \forall t \geq 0.
\end{eqnarray*}
From $P_3(j)$, $P_1(i+1,j), \ldots , P_1(n,j) $, we obtain that  
\begin{eqnarray*}
\abs{ y_i^{(j+1)}(t) }  \leq  \alpha_{\mu_n}  \sum\limits_{l=i+1}^{n} Y_{l,j} + \sum\limits_{q=1}^j  G_{q,j}\overline{\mu}_{n,q} , \quad \forall t \geq I_i.
\end{eqnarray*}
Thus the statement $P_1(j+1,i)$ is proven for all $i \in \llbracket 1, n \rrbracket $.

We now prove by induction on $i$ the statement $P_2(i,j+1)$. As before, since $z_1=y_1$, the case for $i=1$ is done. Assume that for a given $i \in \llbracket 1, n-1 \rrbracket$, the statement $P_2(i_1, j+1)$ holds for all $i_1 \leq i$. From Lemma \ref{lem:techn} (with $k=j+1$, $f=z_i$, $g=y_{i+1}$, $h=z_{i+1}$, $\sigma = \mu_i$, $Q_{k_1}= Z_{i,k_1}$, $M = Y_{i+1,j+1}$, $\overline{\sigma}_a = \overline{\mu}_{i,a}$, $E=I_i$, $F=I_{i+1}$, and \rref{eq:s_i_I_i_prive_I_i+1}), we can establish that $P_2(i+1,j+1)$ holds. $P_2(i,j+1)$ is thus satisfied for all $i \in \llbracket 1, n \rrbracket $.

Finally, we can establish $P_3(j+1)$ from Lemma \ref{lem:techn} (with $k=j+1$, $f=z_n$, $g\equiv 0$, $h=u$, $\sigma = \mu_n$, $Q_{k1}= Z_{n,k_1}$, $M = 0$, $\overline{\sigma}_a = \overline{\mu}_{n,a}$, $E=I_n$, $F=\reels_{\geq 0}$ and \rref{eq:s_n_rn}). This ends the proof of Proposition \ref{prop:2}.
\end{proof}

\vspace{3mm}
We next provide sufficient conditions on the parameters of the saturation functions in $\scal (p)$ guaranteeing global asymptotic stability of the closed-loop system \rref{S1cc}.

\vspace{3mm}
\begin{propo}
\label{prop:stab_S1}
Given $n\in\mathbb N$ and $p\in\mathbb N$, let $ \mu_1 , \ldots , \mu_n$ be saturation functions in $\scal(p)$ with respective constants $(\mu_i^{max},  l_{\mu_i}^{l} , l_{\mu_i}^{s}, \alpha_{\mu_i} ) $ for each $i\in\llbracket 1,n \rrbracket$ and assume that, for all $i \in \llbracket 1 , n-1  \rrbracket$,
\begin{subequations}\label{cond}
\begin{align} 
 \alpha_{\mu_i}  &= 1, \label{cond1} \\
 \mu_i^{max}  &<  L_{\mu_{i+1}} /2. \label{cond2}
\end{align}
\end{subequations}
Then the origin of the closed-loop system \rref{S1cc} is globally asymptotically stable.
\end{propo}
\vspace{3mm}

Actually the above proposition is almost the same as the one given in \cite{Teel92}, except that we allow the first level of saturation $\mu_n$ to have a slope different from $1$. 

\vspace{3mm}
\begin{proof}[Proof of Proposition \ref{prop:stab_S1}]

We prove that after a finite time any trajectory of the closed-loop system \rref{S1cc} enters a region in which the feedback \rref{fe:prop:S_1} becomes simply linear.

To that end, we consider the Lyapunov function candidate $V_n:= \frac{1}{2} y_n^2$. Its derivative along the trajectories of \rref{S1cc} reads
\begin{equation*}
\dot{V}_n= -y_n  \mu_n (y_n + \mu_{n-1}(z_{n-1})).
\end{equation*}
From \rref{cond2}, we can obtain that for all $\abs{y_n} \geq L_{\mu_{n}}/2$
\begin{equation}
\label{lyap}
\dot{V}_n \leq  - \theta L_{\mu_{n}}/2 ,
\end{equation}where $\theta= \inf\limits_{ r \in  [L_{\mu_n}/2 - \mu_{n-1} , S_{\mu_n} ]} \left\lbrace \mu_n (r) \right\rbrace $.

We next show that there exists a time $T_1 \geq 0$ such that $ \abs{y_n(t)} \leq L_{\mu_{n}}/2 $, for all $t \geq T_1$.  To prove that we have the following alternatives : either for every $t \geq 0 $, $\abs{y_n(t)} \leq L_{\mu_n} /2 $ and we are done, or there exist $T_0 \geq 0$ such that $\abs{y_n(T_0)} > L_{\mu_n} /2 $. In that case there exists $ \tilde{T}_0 \geq T_
0 $ such that $y_n(\tilde{T}_0)=L_{\mu_n} /2$ (otherwise thanks to \rref{lyap}, $ V_n(t) \rightarrow - \infty $ as $t \rightarrow \infty $ which is impossible). Due to \rref{lyap}, we have  $\abs{y_n(t)} < L_{\mu_n} /2$ in a right open neighbourhood of $\tilde{T}_0$. Suppose that there exists a positive time $\tilde{T}_1 > \tilde{T}_0$ such that $\abs{y_n(\tilde{T}_1)} \geq L_{\mu_{n}}/2 $. Then by continuity, there must exists $\tilde{T}_2 \in ( \tilde{T}_0 , \tilde{T}_1]$ such that 
$\abs{y_n(\tilde{T}_2)}= L_{\mu_n}/2$, and $\abs{y_n(t)} < L_{\mu_n}/2$ for all $t \in ( \tilde{T}_0 , \tilde{T}_2)$.
However, it then follows from \rref{lyap} that for a left open neighbourhood of $\tilde{T}_2$ we have $\abs{y_n(t)}> \abs{y_n(\tilde{T}_2)}=L_{\mu_n}/2$. This is a contradiction with the fact that on a right open neighbourhood of $\tilde{T}_0$ w have $\abs{y_n(t)} < L_{\mu_n} /2$. 
Therefore, for every $\tilde{T}_1 > \tilde{T}_0$, one has $\abs{y_n(\tilde{T}_1)} < L_{\mu_n} /2$ and the claim is proved.

It follows from \rref{cond2} that
\begin{equation*}
\abs{y_n(t) + \mu_{n-1}(z_{n-1} (t)
)} 
\leq  L_{\mu_n}, \quad \forall t \geq T_1. 
\end{equation*}
Therefore $\mu_n$ operates in its linear region after time $T_1$. Similarly, we now consider $V_{n-1}:= \frac{1}{2} y_{n-1}^2$, whose derivative along the trajectories of \rref{S1cc} satisfies
\begin{equation*}
\dot{V}_n= -  \alpha_{\mu_n}  y_{n-1} \mu_{n-1}\big(y_{n-1}+ \mu_{n-2}(y_{n-2}+ \ldots)\big), \quad \forall t \geq T_1 .
\end{equation*}
Reasoning as before and invoking \rref{cond2}, there exists a time $T_2 >0$ such that $\abs{y_{n-1}(t)} \leq L_{\mu_{n-1}}/2$ and $\mu_{n-1}$ operates in its linear region for all $t \geq T_2$.

By repeating this procedure, we construct a time $T_n$ such that for all times greater than $T_n$ the whole feedback law becomes linear. That is $$\Upsilon(y(t)) = -\alpha_{\mu_n} ( y_n(t) + \ldots + y_1(t) ),$$ for all $t \geq T_n$. System \rref{S1cc} becomes simply linear and its local exponential stability follows readily. Thus the origin of System \rref{S1cc} is globally asymptotically stable, which concludes the proof of Proposition \ref{prop:stab_S1}.
\end{proof}

\subsection{Proof of Theorem \ref{Main_res}}
\label{sec:proof_th}

We now proceed to the proof of Theorem \ref{Main_res} by explicitly constructing the vectors $k_1, \ldots , k_n $ and the constants $a_1 , \ldots , a_n$. This proof can thus be used as an algorithm to compute the nested feedback proposed in Theorem \ref{Main_res}.

Given $p\in\mathbb N$ and $n\in\mathbb N_{\geq 1}$, let $ \sigma_i $ be saturation functions in $\scal (p)$ with constants $( \sigma_i^{max} , L_{\sigma_i} , S_{\sigma_i}, \alpha_{\sigma_i} )$ for each $i \in \llbracket 1 , n \rrbracket$, and let $(R_j)_{0 \leq j \leq p}$ be a family of positive constants. We let 
\begin{align*}
\underline{R} &:= \min \{ R_j : j \in \llbracket 1 , p \rrbracket  \} ,\\
\overline{\sigma}_{n,j} & :=  \max\limits_{r \in  \reels} \left\lbrace \abs{\sigma_n^{(j)}(r)} \right\rbrace, \quad   \forall j \in \llbracket  1, p  \rrbracket ,\\
\alpha_{\tilde{\mu}} & := R_0 L_{\sigma_n} \alpha_{\sigma_n} / \sigma_n^{max} , \\
\tilde{\mu}_{n,j} & :=  \frac{R_0 \overline{\sigma}_{n,j} (L_{\sigma_n})^j }{\sigma_n^{max}},    \quad\forall j \in \llbracket  1, p  \rrbracket ,\\
\overline{b}_{\sigma_n} & := \max \left\lbrace \frac{\sigma_n(r)}{r} : \: 0 < \abs{r} \leq S_{\sigma_n} \right\rbrace , \\
\underline{b}_{\sigma_n} & := \min \left\lbrace \frac{\sigma_n(r)}{r} : \: 0 < \abs{r} \leq S_{\sigma_n} \right\rbrace .  
\end{align*}
Note that all these quantities are well defined since $\sigma_n\in\scal(p)$. We first construct saturations $\mu_1 , \ldots \mu_n$ in order to use results in Section \ref{sec:int:res}. 
Let $(\mu_i^{max})_{1 \leq i \leq n-1}$ and $(L_{\mu_i})_{1 \leq i \leq n-1}$ be two sets of positive constants such that
\begin{equation}
\label{choix_cst_1}
\mu_{n-1}^{max} <  \frac{ 1  }{2}, \quad  L_{\mu_{n-1}}  =   \frac{  \mu_{n-1}^{max} L_{\sigma_{n-1}} \alpha_{\sigma_{n-1}} }{  \sigma_{n-1}^{max}  },
\end{equation} 
and, for each $i \in \llbracket 1, n-2 \rrbracket$,
\begin{equation}
\label{choix_cst_2}
\mu_i^{max} <  \frac{ 1   }{2} L_{\mu_{i+1}} ,\quad  L_{\mu_{i}}  =   \frac{  \mu_{i}^{max} L_{\sigma_{i}} \alpha_{\sigma_{i}} }{  \sigma_{i}^{max}  }.
\end{equation}
For each  $i \in \llbracket 1, n-1 \rrbracket$, the saturation function $\mu_i \in \scal (p)$ with constants ($\mu_i^{max}$, $ L_{\mu_i} $, $S_{\mu_i}$, $1 $), where $ S_{\mu_i} =   S_{\sigma_i}  L_{\mu_i}/ L_{\sigma_i}  $, is then given by
\begin{align*}
\mu_i (s) := \frac{ \mu_i^{max} }{ \sigma_i^{max} } \sigma_i \left( s \frac{ L_{\sigma_i} }{ L_{\mu_i}}   \right),\quad \forall s\in\mathbb R.
\end{align*}
For $\lambda \geq 1$, to be chosen later, we define the saturation function $\mu_n \in \scal (p)$, with constants $\mu_n^{max}=R_0$,  $L_{\mu_n}=\lambda $, $S_{\mu_n}= S_{\sigma_n}  \lambda / L_{\sigma_n}$, and $\alpha_{\mu_n}= \alpha_{\tilde{\mu}} / \lambda$, by
\begin{align*}
\mu_n (s)  :=  \frac{R_0 }{ \sigma_n^{max} } \sigma_n \left( s \frac{ L_{\sigma_n} }{\lambda}   \right), \quad \forall s\in\mathbb R.
\end{align*} 

From \rref{choix_cst_1} and \rref{choix_cst_2} we can establish that the functions $\mu_1 , \ldots , \mu_n $ satisfy conditions \rref{cond}. It follows from Proposition \ref{prop:stab_S1} that the nested feedback law $\Upsilon(y)$ defined in \rref{fe:prop:S_1} stabilizes globally asymptotically the origin of \rref{S1}. 

We next choose $\lambda$ in such a way that $\Upsilon(y)$ is a $p$-bounded feedback law by $(R_j)_{0 \leq j \leq p}$ for System \rref{S1}. To that end, first notice that 
\begin{eqnarray}
\overline{b}_{\mu_n} & = & \frac{ \alpha_{\tilde{\mu}}  \overline{b}_{\sigma_n} }{\lambda \alpha_{\sigma_n}} , \label{pr:th:lab1} \\
 \underline{b}_{\mu_n} & = & \frac{ \alpha_{\tilde{\mu}} \underline{b}_{\sigma_n} }{\lambda \alpha_{\sigma_n}}, \label{pr:th:lab2} \\
\underline{B}_{\mu_n} & = & \frac{1}{\lambda} \min \left\lbrace \frac{ \alpha_{\tilde{\mu}}  \underline{b}_{\sigma_n} }{\alpha_{\sigma_n}} , \: \frac{R_0}{ S_{\sigma_n} / L_{\sigma_n} + 2 \sigma_{n-1}^{max} / \lambda} \right\rbrace , \label{pr:th:lab3} \\
\overline{\mu}_{n,q} & = &  \tilde{\mu}_{n,q} / \lambda^q , \label{pr:th:lab4}
\end{eqnarray}
where $\overline{b}_{\mu_n}$, $\underline{b}_{\mu_n}$, and $\overline{\mu}_{n,q}$ are defined in \rref{def_bound_mu_n_sup},  \rref{def_bound_mu_n_inf}, and \rref{def_up_bnd_et} respectively. Using Proposition \ref{prop:2}, it follows that every trajectory of the closed -loop system \rref{S1cc} satisfies, for each $j\in\llbracket 1, p\rrbracket$,
\begin{equation}
\label{eq:uj_est_pr_th1}
\sup\limits_{t\geq 0} \left\lbrace \abs{u^{(j)}(t)} \right\rbrace \leq   \sum\limits_{q=1}^j   G_{q,j} \frac{\tilde{\mu}_{n,q} }{\lambda^q }.
\end{equation}

By substituting \rref{pr:th:lab1}, \rref{pr:th:lab2}, \rref{pr:th:lab3}, and \rref{pr:th:lab4} into the recursion in Proposition \ref{prop:2}, it can be seen that, for each  $j\in \llbracket  1,p \rrbracket $, $\sum\limits_{q=1}^j  \frac{ G_{q,j} \tilde{\mu}_{n,q} }{\lambda^q }= \frac{1}{\lambda} P( \frac{1}{\lambda} )$ where $P$ is a polynomial with positive coefficients. This sum is thus decreasing in $\lambda$. Hence, we can pick $\lambda\geq 1$ in such a way that
\begin{equation*}
 \sum\limits_{q=1}^j  \frac{ G_{q,j} \tilde{\mu}_{n,q} }{ \lambda^q } \leq \underline{R}, \quad\forall j \in \llbracket 1, p \rrbracket.
\end{equation*}
It follows that, for each  $j\in\llbracket1 , p\rrbracket$, 
\begin{equation*}
\sup\limits_{t\geq 0} \left\lbrace \abs{u^{(j)}(t)} \right\rbrace \leq   \underline{R} \leq R_j.
\end{equation*}
Recalling that the feedback $\Upsilon$ is bounded by $R_0$, we conclude that it is $p$-bounded feedback law by $(R_j)_{0 \leq j \leq p}$ for System \rref{S1}.

With the linear change $y=Hx$, the closed-loop system \rref{S1cc} can be put into the form of the closed-loop system \rref{mult_int} with $u=\Upsilon(Hx)$. Thus, the sought feedback law $\nu$ of Theorem \ref{Main_res} is obtained by $\nu(x)=\Upsilon(Hx)$. This leads to the following choices of parameters:
\begin{align*}
a_n &= R_0 / \sigma_n^{max},\\
a_i & = \frac{L_{\sigma_{i+1}} \mu_i^{max} }{L_{\mu_{i+1}}  \sigma_i^{max}}, \quad \forall i \in  \llbracket 1, n-1 \rrbracket , \\
k_n^T x &=  \frac{L_{\sigma_n} }{ L_{\mu_{n}} } x_n, \\
k_{n-i}^T x &=  \frac{L_{\sigma_{n-i}}}{L_{\mu_{n-i}} } \sum\limits_{k=0}^i \frac{i!}{k!(i-k)!} \left( \frac{\alpha_{\tilde{\mu}}}{ L_{\mu_{n}} } \right)^k  x_{n-k}, \quad \forall i \in  \llbracket 1, n-1 \rrbracket ,
\end{align*}
and $L_{\mu_{n}}=\lambda$.

\section{Simulation}

In this section, we illustrate the applicability and the performance of the proposed feedback on a particular example. We use the procedure described in Section \ref{sec:proof_th} in order to compute a $2$-bounded feedback law by $(2,20,18)$ for the multiple integrator of length three. Our set of saturation functions is $\sigma_1 = \sigma_2 = \sigma_3 = \sigma$ where $\sigma$ is an $\scal(2)$ saturation function with constants $(2,1,2,1)$ given by
\begin{equation*}
\sigma (r) := \left\lbrace \begin{array}{l l}
r &   \quad \text{if } \abs{r} \leq 1, \\
h_1(r) &\quad \text{if } 1 \leq \abs{r} \leq 1.5,\\
h_2(r) &  \quad \text{if } 1.5 \leq \abs{r} \leq 2,\\
2  \sign (r) & \quad  \text{otherwise,} \
\end{array} \right.
\end{equation*}
with $h_1$ and $h_2$ were picked in order to ensure sufficient smoothness for $\sigma$: 
\begin{align*}
h_1(r) &:= \sign (r) ( -4 + 15 \abs{r} - 18 r^2 + 10 \abs{r}^3 - 2r^4 ),\\
h_2(r)&:= 2 \sign (r) ( 25 - 60 \abs{r} + 54 r^2 -21 \abs{r}^3 + 3 r^4 ).
\end{align*}
In accordance with \rref{choix_cst_1} and \rref{choix_cst_2}, we choose $\mu_2^{max} = 2/5$, $L_{\mu_2} = 1/5$, $\mu_1^{max} = 1/12$, and $L_{\mu_1} = 1/24$. Following the procedure, we obtain that 
\begin{eqnarray*}
\sup\limits_{t\geq 0} \left\lbrace \abs{u^{(1)}(t)} \right\rbrace & \leq &(7.91 + 4.35 \lambda )/ \lambda^2 , \\
\sup\limits_{t\geq 0} \left\lbrace \abs{u^{(2)}(t)} \right\rbrace & \leq  & \frac{26.2 \lambda^3 + 396 \lambda^2+1147.2 \lambda +125.2}{\lambda^4}.
\end{eqnarray*}
Choosing $\lambda = 6.5$, we obtain that $\sup\limits_{t\geq 0} \left\lbrace \abs{u^{(1)}(t)} \right\rbrace  \leq 0.9 $, and  $\sup\limits_{t\geq 0} \left\lbrace \abs{u^{(2)}(t)} \right\rbrace  \leq 18 $. The desired feedback is then given by
\begin{align*}
\nu(x) = &  - \sigma \Big( \frac{1}{6.5} \Big( x_3 + \frac{1}{5} \sigma \big(5 ( x_2/6.5 + x_3 \\
&  + \frac{1}{24} \sigma \big( 24 ( x_3 + 2x_2/6.5 + x_1/6.5^2 )) \big)\big)\Big)\Big).
\end{align*}

This feedback law was tested in simulations. The results are presented In Figure \ref{figua}. Trajectories of the multiples integrator of length $3$ with the above feedback are plotted in grey for several initial conditions. The corresponding values of the control law and its time derivatives up to order $2$ are shown in Figure \ref{figub}. These grey curves validate the fact that asymptotic stability is reached and that the control feedback magnitude, and two first derivatives, never overpass the prescribed values $(2,20,18)$. In order to illustrate the behaviour of one particular trajectory, the specific simulations obtained for initial condition  $x_{10}=446.7937$, $x_{20} = -69.875$ and $x_{30}=11.05$ are highlighted in bold black.

It can be seen from Figure \ref{figub} that our procedure shows some conservativeness the amplitude of the second derivative of the feedback never exceeds the value $2$, although maximum value of $18$ was tolerated.
\label{sec:simu}
\begin{figure}[thpb]
      \centering
      \includegraphics[height=15cm, width=\textwidth]{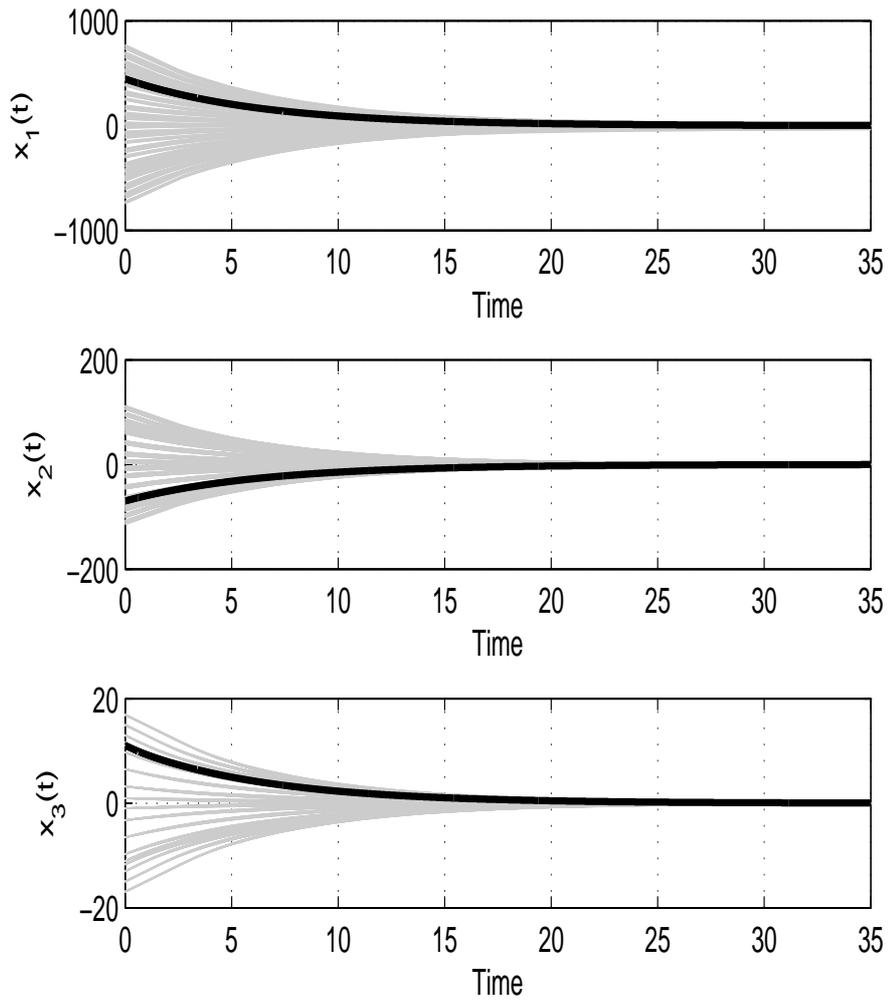}
      \caption{Evolution of the states for a set of initial conditions.}
      \label{figua}
      \end{figure}
      
\begin{figure}[thpb]
      \centering
      \includegraphics[height=15cm, width=\textwidth]{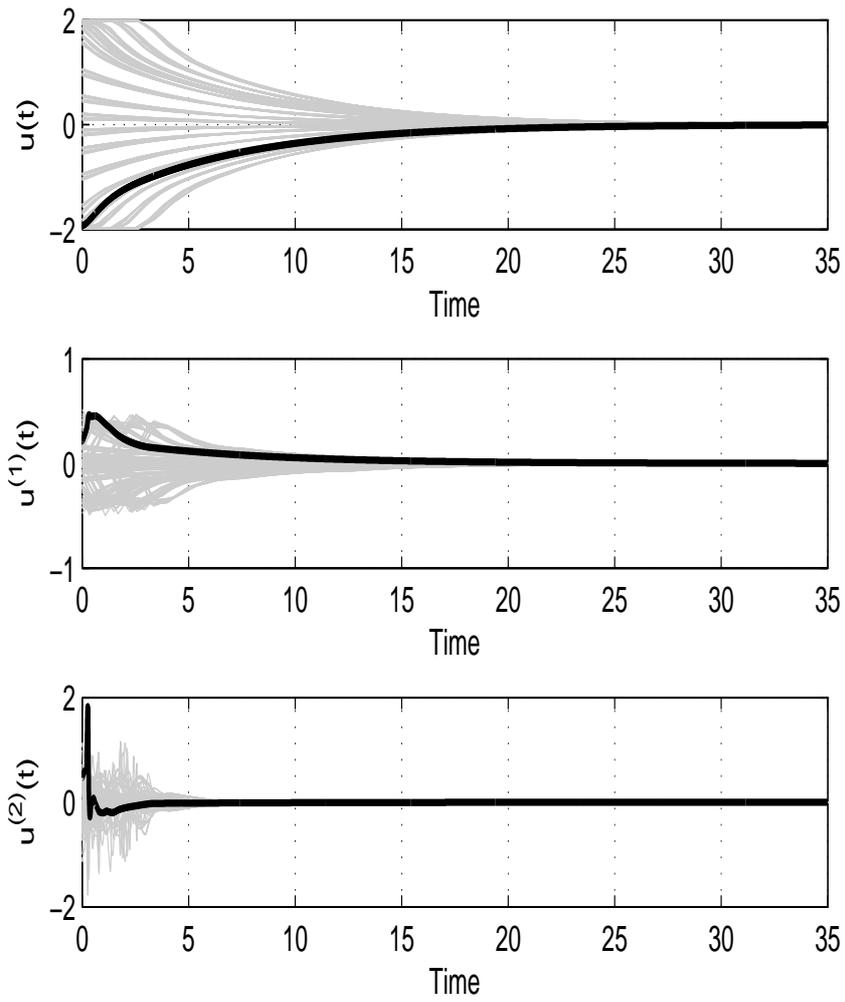}
      \caption{Evolution of the control and its derivative up to order 2 for the same set of initial conditions.}
      \label{figub}
\end{figure}

\section{Conclusion}\label{sec: conclusion}

We have shown that any chain of integrators can be globally asymptotically stabilized by a static feedback whose magnitude and $p$ first time derivatives are below arbitrary prescribed values, uniformly with respect to all trajectories of the closed loop system. The design of this feedback
relies on the technique of nested saturations first introduced in \cite{Teel92}. The applicability of the design procedure and the performance of the resulting closed-loop system was tested on a particular example.

The following two problems can be considered for future works: $i)$ extending this result from integrator chains to general linear systems stabilizable by bounded input; $ii)$ designing a $C^\infty$ bounded feedback, all the successive derivatives of which stand below prescribed constants at all times.

\bibliographystyle{IEEEtran}
\bibliography{IEEEabrv,biblio}
\end{document}